\documentclass{amsart}
\usepackage{latexsym}
\usepackage{pict2e}
\usepackage{epsfig}
\usepackage{amsmath,amsthm,amssymb}
\usepackage{graphicx}
\usepackage{graphics}
\usepackage{algorithmicx,algorithm}
\usepackage{algpseudocode}
\usepackage{indentfirst}
\usepackage{listings,multicol}
\usepackage{lipsum}
\usepackage{xcolor}
\usepackage{upquote}
\usepackage[T1]{fontenc}
\usepackage[hang]{footmisc}
 \newtheorem{definition}{Definition}
\newtheorem{lemma}{Lemma}

\newtheorem{theorem}{Theorem}
\newtheorem{corollary}{Corollary}

\usepackage{cite}
\usepackage{multirow}
\usepackage{enumerate}
\newcommand{\dcell}{\text{DCell}}
\newcommand{\dcc}[3]{D^{#3}_{#2}}
\newcommand{\imgsize}{0.45}
\newcommand{\largeimgsize}{0.85}
\newcommand{\HPSeq}{\texttt{HPSequence}}
\newcommand{\DCellHP}{\texttt{DCellHP}}
\newcommand{\adddcell}{\texttt{AddDCell}}

\RequirePackage{graphicx}

\title{Hamiltonian Properties of DCell Networks}

\author{Xi Wang}
\address{School of Computer Science and Technology, Soochow University, Suzhou 215006, China}
\author{Alejandro Erickson}
\address{School of Engineering and Computing Sciences, Durham University, DH1 3LE, United Kingdom}
\author{Jianxi Fan}
\address{School of Computer Science and Technology, Soochow University, Suzhou 215006, China}
\author{Xiaohua Jia}
\address{Department of Computer Science, City University of Hong Kong, Kowloon, Hong Kong}

\begin{document}
\maketitle

\begin{abstract}
DCell has been proposed for data centers as a server centric interconnection network structure.
DCell can support millions of servers with high network capacity by only using commodity switches.
With one exception, we prove that a $k$ level DCell built with $n$
port switches is Hamiltonian-connected for $k \geq 0$ and $n \geq 2$.
Our proof extends to all generalized DCell connection rules for $n\ge 3$.
Then, we propose an $O(t_k)$ algorithm for finding a Hamiltonian
path in $DCell_{k}$, where $t_k$ is the number of servers in
$DCell_{k}$.
What's more, we prove that $DCell_{k}$ is $(n+k-4)$-fault Hamiltonian-connected and
$(n+k-3)$-fault Hamiltonian.
In addition, we show that a partial DCell is Hamiltonian connected if it conforms to a few
practical restrictions.

\end{abstract}

\section{Introduction}
Data centers are critical to the business of companies such as Amazon,
Google, Facebook, and Microsoft.  These and other corporations operate
data centers with hundreds of thousands of servers.
Their operations are important to offer both many on-line applications such as web search, on-line gaming, email,
cloud storage, and infrastructure services such as GFS \cite{ghemawat2003google}, Map-reduce \cite{dean2008mapreduce}, and Dryad \cite{isard2007dryad}.
The growth in demand for such services has lately exceeded the growth
in performance afforded by existing data center technology and
topology.  In particular, we are faced by the challenge of
interconnecting a large number of servers in one data center, at a low
cost, and without compromising performance.  Toward this end, Al-Fares
et al. \cite{al2008scalable} introduced the idea of replacing high-end
networking hardware with commodity switches in a data center network
called Fat-tree.  Concurrently, Guo et al. proposed a server-centric
data center network called DCell \cite{guo2008dcell}, wherein the
commodity switches have no intelligence at all, and all of the routing
intelligence is restricted to the servers.


DCell is particularly apt to handle a very large number of servers,
even on the order of millions.  It scales double exponentially in the
number of ports in each server, it has high network capacity, large
bisection width, small diameter, and high fault-tolerance.  DCell only
requires mini-switches and can support a scalable routing algorithm.


DCell is part of a class of data center network designs that evolved
from parallel computing interconnects (see \cite{LinLiuHamdi2012}),
where Hamiltonian cycles and paths are commonly used for making low
congestion and deadlock-free message broadcasts (e.g
\cite{lin1994deadlock,wang2005multicast}).  Although the applications
and traffic of high performance parallel computing and data center
networks differ from each other, broadcasts are likely to be used in a
data center in order to update information about the network, perform
distributed computations, etc.  For example, broadcasting is
implemented both in DCell and BCube \cite{GuoLuLi2009}.  It is natural
to consider Hamiltonian broadcast schemes in such server-centric data
centres, especially when slower broadcasts are acceptable and
bandwidth conservation is critical.

It is well known that there is no nontrivial necessary and sufficient
condition for a graph to be Hamiltonian, and that finding a
Hamiltonian cycle or path is NP-Complete \cite{garey1979computers,johnson1982np}.  Therefore, a large
amount of research on Hamiltonicity focuses on different special
networks.
Fan showed that the $n$-dimensional M\"{o}bius cube $M_n$ is Hamiltonian-connected when $n\geq3$ \cite{fan2002hamilton}.
Park et al. showed that every restricted hypercube-like interconnection networks of degree $m$$(m \geq 3)$ is $(m - 3)$-fault Hamiltonian-connected and $(m - 2)$-fault Hamiltonian \cite{park2005fault}.
Let $G$ be an $n$-dimensional twisted hypercube-like networks with $n \geq 7$ and let $F$ be a subset of $V(G) \cup E(G)$ with $|F| \leq 2n-9$.
Yang et al. proved that $G \backslash F$ contains a Hamiltonian cycle if $\delta(G \backslash F) \geq 2$ \cite{yang2011hamiltonian}.
Wang treated the problem of embedding Hamiltonian cycles into a crossed cube with failed links and found a Hamiltonian cycle in a crossed cubes $CQ_n$ tolerating up to $n-2$ failed links \cite{wang2012hamiltonian}.
Xu et al. provided a systematic method to construct a Hamiltonian path in Honeycomb meshes \cite{xu2013hamiltonian}.



A flood-like broadcast scheme for DCell, called DCellBroadcast, is
used in \cite{guo2008dcell} instead of a tree-like multicast, because
it is fault tolerant.  DCellBroadcast creates congestion because the
broadcast message is replicated many times, and for bandwidth-critical
applications it is worth revisiting the broadcast problem.  In doing
so, we explore the avenue of Hamiltonian cycle or path based multicast
routing, inspired by those in parallel computing (\cite{lin1994deadlock,wang2005multicast}).


So far, there is no work reported about the Hamiltonian properties of
DCell.  The major contributions of this paper are as follows:

\begin{enumerate}[(1)]
\item We prove that a $k$ level DCell built with $n$ port switches is
Hamiltonian-connected for $k \geq 0$ and $n \geq 2$, except for
$(k,n)=(1,2)$,
\item we prove that a $k$ level Generalized DCell
  \cite{KlieglLeeLi2009} built with $n$ port switches is
  Hamiltonian-connected for $k \geq 0$ and $n \geq 3$,
\item we propose an $O(t_k)$ algorithm for finding a Hamiltonian path
  in $DCell_{k}$, where $t_k$ is the number of servers in $DCell_{k}$,

\item we prove that a $k$ level $DCell$ with up to $(n+k-4)$-faulty
  components is Hamiltonian-connected and with up to $(n+k-3)$-faulty
  components is Hamiltonian, and,

\item we prove that a partial DCell is Hamiltonian-connected if it
  conforms to a few practical restrictions.
\end{enumerate}

This work is organized as follows.
Section~\ref{sec:preliminaries} provides the preliminary knowledge.
Our main result, that DCell is Hamiltonian-connected, is given in Section~\ref{sec:hc-dcell}.
Section~\ref{sec:fault-tolerance} discusses fault-tolerant Hamiltonian properties of DCell.
Hamiltonian properties in partial DCells are given in Section~\ref{sec:incremental}.
We provide some discussions in Section~\ref{sec:discussion}.
We make a conclusion in Section~\ref{sec:conclusion}.

\section{Preliminaries }
\label{sec:preliminaries}

Let a data center network be represented by a simple graph $G = \left(V\left(G\right),E\left(G\right)\right)$,
where $V\left(G\right)$ represents the vertex set and $E\left(G\right)$ represents the edge set, and
each vertex represents a server and each edge represents a link between servers (switches can be regarded as transparent network devices \cite{guo2008dcell}). The edge between vertices $u$ and $v$ is denoted by $(u, v)$. In this paper all graphs are simple and undirected.

A $(u_0,u_n)$-path of length $n$ in a graph is a sequence of vertices, $P:(u_0, u_1, \dots, u_j,\dots u_{n-1}, u_{n})$,
in which no vertices are repeated and $u_j, u_{j+1}$ are adjacent for any integer $0 \leq j < n$.
We use $V(P)$ and $E(P)$ to denote the vertices and edges of $P$, respectively.
If $u$ and $v$ are vertices on a path $Q$, we write $P(Q,u,v)$ to denote the sub-path of $Q$ from $u$ to $v$.
If $P_2$ contains only the edge
$(v,w)$, we simply write $P_1+(v,w)$, and furthermore, we allow the
subtractive analog, so that $(P_1+(v,w))-(v,w) = P_1$.  

A $(u,v)$-path in a graph $G$ containing every vertex of $G$ is called
a Hamiltonian path, and it is denoted $HP(u, v, G)$.
If $(v, u)\in E(G)$, then $HP(u, v, G) + (u,v)$ is a
Hamiltonian cycle $C$. Thus, we say $C - (u,v)$ is a Hamiltonian path $HP(u, v, G)$. A Hamiltonian graph is a graph containing a
Hamiltonian cycle. 
If there exists a Hamiltonian path between any two distinct vertices of
$G$, then $G$ is  Hamiltonian-connected.  It is easy to see that if $G$ is a Hamiltonian-connected graph with $|V(G)|
\geq 3$, then $G$ must be a Hamiltonian graph.


For undefined graph theoretic terms see \cite{Diestel2012}.

DCell uses a recursively defined structure to interconnect servers. Each server connects to different levels of DCell through multiple links.
We build high-level DCell recursively to form many low-level ones.

According to the definition of $DCell_{k}$ \cite{guo2008dcell}, we provide the recursive definition as Definition~\ref{def:dcell}.

\begin{definition}[DCell, \cite{guo2008dcell}]
  \label{def:dcell}
 Let $D_k$ denote a level $k$ DCell, for each $k\ge 0$ and some global
  constant $n$.  Let $D_0 = K_n$, and let $t_k$ denote the number of
  vertices in $D_k$ (thus $t_0=n$).

  For $k>0$, the graph $D_{k}$ is built from $t_{k-1} +1$ disjoint
  copies of $D_{k-1}$, where $D^i_{k-1}$ denotes the $i$th copy.  Each
  pair of DCells, $(D^a_{k-1},D^b_{k-1})$, is connected by a
  \emph{level $k$ edge}, $(u,v)$, according to the rule described in
  \textbf{Connection rule}, below.  We say that $v$ is the (unique)
  level $k$ neighbor of $u$.

A vertex $x$ (of $D_k$) in $D_{k-1}^i$ is labeled
$(i,\alpha_{k-1},\ldots,\alpha_{0})$ where $k>0$, and $\alpha_0\in \{1,2,\ldots, n\}$.

The suffix, $(\alpha_j,\alpha_{j-1}, \ldots,\alpha_0)$, of the label
$\alpha$, has the unique uid$_j$, given by $uid_j(\alpha)=\alpha_{0} +
\sum_{l=1}^{j} (\alpha_{l} t_{l-1})$. In $D_k$, each vertex is
uniquely identified by its full label, or alternatively, by a $uid_j$
and the corresponding prefix of the label.

\begin{description}
\item[Connection rule:] For each pair of level $k-1$ DCells, $(D^a_{k-1},D^b_{k-1})$, with $a < b$,
vertex uid$_{k-1}$ $b-1$ of $D^a_{k-1}$ is connected to vertex uid$_{k-1}$ $a$ of $D^b_{k-1}$.
\end{description}
\end{definition}

\begin{figure*}
\centering
\includegraphics[width=\largeimgsize\textwidth]{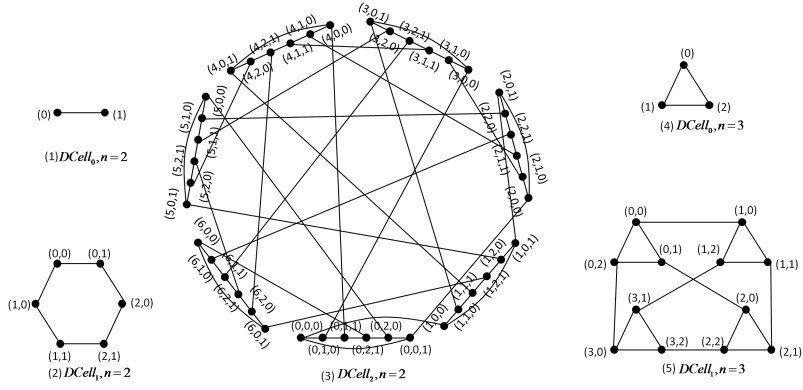}
\caption{Small DCells and some special cases with regards Hamiltonian
  connectivity.}
\label{fig:0}
\end{figure*}

Figure~\ref{fig:0} depicts several DCells with small parameters $n$ and $k$.

Definition~\ref{def:dcell} generalizes by replacing the connection
rule with a different one that satisfies the requirement that each
vertex be incident with exactly one level $k$ edge (see \cite{guo2008dcell}
).  We have indicated where our results apply to Generalized
DCell, but $D_k$ and $DCell_k$ refers to the connection rule in
Definition~\ref{def:dcell} unless otherwise stated.


\section{Hamiltonian Connectivity of DCell}
\label{sec:hc-dcell}

In this section we prove by induction on $k$, that $\dcell_k$ is
Hamiltonian Connected in all but a few inconsequential cases.  The base
cases are given as Lemmas~\ref{lem:1}-\ref{lem:4}, and the main result is Theorem~\ref{thm:1}.  We
indicate wherever the results also hold for generalized DCell, but
$D_k$ and $DCell_k$ continue to refer to the graph in
Definition~\ref{def:dcell}.
Furthermore, we propose an $O(t_k)$ algorithm for finding a Hamiltonian path in $DCell_{k}$, where $t_k$ is the number of servers in $DCell_{k}$.


\begin{lemma}
  \label{lem:1}
For any integer $n$ with $n \geq 2$, (Generalized) $DCell_{0}$ is Hamiltonian-connected.
\end{lemma}
\begin{proof}
 The lemma holds, since $DCell_{0}$ is a complete graph \cite{Diestel2012}.
\end{proof}

\begin{lemma}
  \label{lem:2}
  $DCell_{1}$ is a Hamiltonian graph with $n = 2$. However,
  $DCell_{1}$ is not Hamiltonian-connected with $n = 2$.  This also holds for Generalized DCell.
\end{lemma}
\begin{proof}
  $DCell_{1}$ is a cycle on 6 vertices. Therefore, for $n = 2$,
  $DCell_{1}$ is a Hamiltonian graph, and $DCell_{1}$ is not
  Hamiltonian-connected \cite{Diestel2012}.
\end{proof}

\begin{lemma}
  \label{lem:3}
$DCell_{2}$ is Hamiltonian-connected with $n = 2$. This does not hold for all Generalized DCell, specifically, not for $\beta$-DCell in \cite{kliegl2010generalized}.
\end{lemma}
\begin{proof}
  For $n = 2$, we find a $(u,v)$-Hamiltonian path for every pair of
  vertices in $\dcell_2$ using a computer program. On the other hand, the $\beta$-DCell from
  \cite{kliegl2010generalized} fails this test.
\end{proof}

The negative result for Generalized DCell in Lemma~\ref{lem:3} appears
to weaken Theorem~\ref{thm:1}, but the case where $n=2$ is
inconsequential, since no reasonable Generalized DCell would be
constructed with $2$-port switches.

\begin{lemma}
  \label{lem:4}
  $DCell_{1}$ is Hamiltonian-connected with $n = 3$.  This also holds for Generalized DCell.
\end{lemma}
\begin{proof}
  This is also verified by a computer program, and the observation that all Generalized DCell with these parameters are isomorphic.
\end{proof}

\begin{theorem}
  \label{thm:1}
For any integer $n$ and $k$ with $n \geq 2$ and $k \geq 0$, $DCell_{k}$ is Hamiltonian-connected, except for $DCell_{1}$ with $n = 2$.  Generalized DCell is Hamiltonian connected for $n\ge 3$ and $k\ge 0$.
\end{theorem}
\begin{proof}
  We proceed by induction on the dimension, $k$, of $\dcell_k$.  The
  base cases are given in Lemmas~\ref{lem:1}--\ref{lem:4}, and we
  prove an induction step which holds for Generalized DCell.

  Let $D_{k}$ denote $DCell_k$ with $n$-port switches.
  Our induction hypothesis is that $D_{k-1}$ is Hamiltonian-connected for $k> 0$ when $n\ge 3$, and $k>2$ when $n=2$.

The graph $D_{k}$ is built from $t_{k-1}+1$ copies of $D_{k-1}$, and every pair of distinct
$D_{k-1}$s is connected by exactly one $k$-\emph{level edge}.
A specific copy of $D_{k-1}$ is denoted by $D_{k-1}^\alpha$, with
$\alpha \in \{0,1,\cdots,t_{k-1}\}$.  Therefore, we can think of
$D_{k}$ as a complete graph whose vertices are the $D_{k-1}$s and
whose edges are the level $k$ edges of $D_k$.  Let $G=(V,E)$ be the
graph that is isomorphic to the complete graph $K_{t_{k-1}+1}$, with
$V=\{0,1,\ldots, t_{k-1}\}$, where vertex $i$ corresponds to
$D_{k-1}^i$, and edge $(i,j)$ corresponds to the level $k$ link that
connects $D_{k-1}^i$ to $D_{k-1}^j$.  We combine the level $k$ edges corresponding to Hamiltonian cycles
and paths in $G$ with the Hamiltonian paths of $D_{k-1}$ to prove the
induction step.

Our goal is to prove that there is a Hamiltonian path between any pair of distinct
vertices, $u,v \in V(D_{k})$, and we consider three cases: Either $u$ and $v$ are in the
same copy of $D_{k-1}$,
or else they are in distinct copies of $D_{k-1}$.
In this case,
either $(u,v) \in E(D_{k})$ or not.

Case 1, $u$ and $v$ are in the same copy of $D_{k-1}$. Let
$u, v \in V(D_{k-1}^{\alpha})$, with $u \neq v$. There is a
$(u,v)$-Hamiltonian path, $P$, where $x$ is a adjacent to $v$ on $P$.
Let $D_{k-1}^{\beta}$ and $D_{k-1}^{\gamma}$ be the distinct subgraphs
connected to $x$ and $v$, respectively.  Let $H_G$ be a Hamiltonian
cycle in $G$, which contains the path $(\beta,\alpha,\gamma)$, and let
$H$ be the corresponding set of level $k$ edges in $D_{k}$.  By the
induction hypothesis there is a Hamiltonian path in each $D_{k-1}^{i}$
for $i \neq \alpha$ whose first and last vertices are adjacent to
$k$-level edges in $H$. The union of these paths with $H$ and $P -
(x,v)$ is a $(u, v)$-Hamiltonian path (refer to Figure~\ref{fig:1}).
  \begin{figure}
    \centering
    \def\svgwidth{\imgsize\textwidth}
{\small    \input{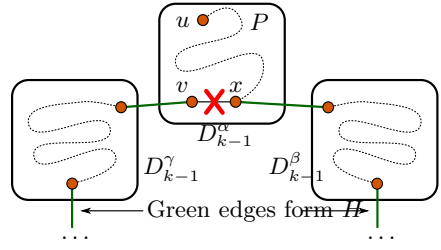}}
\caption{Case 1. Here, $H$ corresponds to a Hamiltonian cycle $G$
  containing the path $(\beta,\alpha,\gamma)$.}
    \label{fig:1}
  \end{figure}

  Case 2, $u$ and $v$ are in distinct copies of $D_{k-1}$.  Let
  $u \in V(D_{k-1}^{\alpha})$ and $v \in V(D_{k-1}^{\beta})$, with
  $\alpha \neq \beta$. There are two sub-cases.

  Case 2.1, $(u, v) \in E(D_{k})$. Let $H_G$ be a Hamiltonian
  cycle in $G$ which contains the edge $(\alpha,\beta)$, and let $H$
  be the corresponding set of level $k$ edges in $D_k$. The union of
  $H$ with the Hamiltonian paths in each copy of $D_{k-1}$, minus $(u,
  v)$, is the required $(u, v)$-Hamiltonian path (refer to
  Figure~\ref{fig:2.1}).

    \begin{figure}
      \centering
      \def\svgwidth{\imgsize\textwidth}
{\small      \input{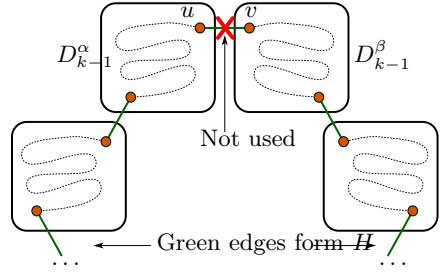}}
      \caption{Case 2.1. Here, $H$ corresponds to a Hamiltonian cycle in
        $G$ containing the edge $(\alpha,\beta)$, where
        $(\alpha,\beta)$ corresponds to $(u,v)$ in $D_k$.}
      \label{fig:2.1}
    \end{figure}

    Case 2.2, $(u,v) \not\in E(D_{k})$.  Let $u$ and $v$ be
    connected by $k$-level edges to vertices in $D_{k-1}^{\gamma}$ and
    $D_{k-1}^{\delta}$, respectively. Note that it is possible to have
    $\gamma = \delta$.  Let $G'$ be the graph $G$ minus the set of
    edges $\{(\alpha,\gamma),(\beta, \delta)\}$.  There is an
    $(\alpha,\beta)$-Hamiltonian path in $G'$, because $t_{k-1} + 1
    \geq 5$, so let $H_G$ be this path, and let $H$ be the
    corresponding set of level $k$ edges in $D_k$. The union of $H$
    with the appropriate paths obtained by the induction hypothesis is
    the required $(u, v)$- Hamiltonian path (refer to
    Figure~\ref{fig:2.2}).

    \begin{figure}
      \centering
      \def\svgwidth{\imgsize\textwidth}
{\small      \input{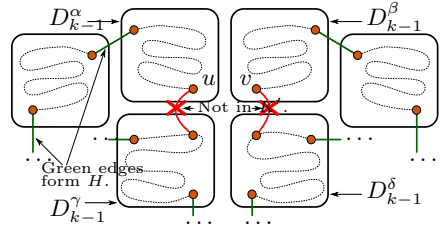}}
\caption{Case 2.2. Here, $H$ corresponds to an
  $(\alpha,\gamma)$-Hamiltonian path in $G'$.}
      \label{fig:2.2}
    \end{figure}
\end{proof}

Theorem~\ref{thm:1} converts readily into an algorithm, which we give
as Algorithm~\ref{alg:dcellhp}.  In order to express the algorithm
compactly, we need some notation.  Let $x\in V(D_{k})$, and let $0< j
\le k$.  The (unique) level $j$ neighbor of $x$ is denoted $N(x,k)$.
Let $D_{k-1}^\alpha$ and $D_{k-1}^\beta$ be distinct $D_{k-1}$s in
$D_k$.  The (unique) level $k$ edge which connects $D_{k-1}^\alpha$ to
$D_{k-1}^\beta$ is denoted $e(D_{k-1}^\alpha,D_{k-1}^\beta)$.  Let $A$
be a set.  Denote the permutation $(\sigma_0,\sigma_1,\ldots,
\sigma_{-1})$ of the elements of $A$ by $\sigma(A)$, where
$\sigma_{-1}$ denotes the last element in the permutation.

\algnewcommand\AND{\textbf{and~}}
\algnewcommand\OR{\textbf{or~}}

\begin{algorithm}
  \caption{\DCellHP~returns a $(u,v)$-Hamiltonian Path in $D_{k}$.
    Note that $u=(u_K,u_{K-1},\ldots,u_0)$ and $v=(v_K,v_{K-1},\ldots,
    v_0)$, for some $K\ge k$, and \DCellHP$(u,v,k,n)$ operates only on
    the length $k+1$ suffixes of $u$ and $v$.  Recall that if $k>0$,
    then $u$ is in $D_{k-1}^{u_k}$.}
  \label{alg:dcellhp}
  \begin{algorithmic}
    \Require Vertices $u,v$ in $V(D_{k})$ are distinct.  If $k=1$,
    then $n\neq 2$. %
    \Function{\DCellHP}{$u,v,k,n$} %
    \Comment{See proof of Theorem~\ref{thm:1}.}
    \If{$k = 0$ \OR ($k = 2$ \AND $n = 2$) \OR ($k = 1$ \AND $n =
      3$)}
    \Comment{Lemmas~\ref{lem:1}--\ref{lem:4}.}

    \State\Return A Hamiltonian Path from $u$ to $v$.
    \EndIf %
    \If{$u_k = v_k$} \Comment{Case 1.}
    \State $P' \gets $\DCellHP$(u,v,k - 1,n)$.
    \State Choose $x \in V(P')$ such that $(x,v) \in E(P')$.
    \State $x' \gets N(x,k)$ and $v' \gets N(v,k)$.
    \State $P \gets P'-(x,v) + (x,x')$.
    \State Compute $\sigma(\{0, 1, \cdots, t_{k-1}\} \setminus
    \{u_k,v'_k,x'_k\})$.
    \State $H \gets (x'_k,\sigma_0,\sigma_1,\ldots, \sigma_{-1},
    v'_k)$.
    \State \Return{$P + $ \HPSeq$(H,k,n,u',v')$$+(v',v)$}.
    \Else \Comment{Case 2.1 and 2.2.}
    \State $u' \gets N(u,k)$ and $v' \gets N(v,k)$.
    \State Compute $\sigma(\{0, 1, \cdots, t_{k-1}\} \setminus
    \{u_k,v_k\})$ such that $\sigma_0\neq u'_k$ and $\sigma_{-1}\neq
    v'_k$.
    \State $H \gets (u_k,\sigma_0,\sigma_1,\ldots, \sigma_{-1}, v_k)$.
    \State\Return{\HPSeq $(H,k,n,u,v)$}.
    \EndIf%
    \EndFunction%

    \Function{\HPSeq}{$H,k,n,u,v$}%
    \State\Comment{$H=(H_0,H_1, \ldots, H_{-1})$.}
    \If{$u_k = v_k$}%
    \State\Return{\DCellHP$(u,v,k-1,n)$}.%
    \EndIf%
    \State $(x,x')\gets e(D_{k-1}^{H_0},D_{k-1}^{H_1})$.
    \State\Return{\DCellHP$(u,x,k-1,n) + (x,x') +$
      \HPSeq$((H_1,H_2,\ldots, H_{-1}),k,n,x',v)$}.
    \EndFunction %
  \end{algorithmic}
\end{algorithm}

We prove the running time of Algorithm~\ref{alg:dcellhp} in
Theorem~\ref{thm:2}.

\begin{theorem}
  \label{thm:2}
  There exists an $O(t_k)$ algorithm for finding a $(u,v)$-Hamiltonian
  path in $\dcell_k$.
\end{theorem}

\begin{proof}
  Algorithm~\ref{alg:dcellhp} returns a $(u,v)$-Hamiltonian path.  We
  assume, in our time analysis, that $N(x,k)$ and
  $e(D_{k-1}^\alpha,D_{k-1}^\beta)$ can be computed in constant time,
  which is the case when using the connection rule given in
  Definition~\ref{def:dcell}.

  The operations in \DCellHP, save for recursive calls, and calls to
  \HPSeq, can be performed in constant time.  In particular, the
  permutations can be selected from pre-computed permutations of
  length $t_{k-1}+1$ by skipping the appropriate elements.  There are
  $t_{k-1} +1$ calls to \DCellHP$(\cdot,\cdot,k-1,n)$, including those
  in \HPSeq, with a constant amount of overhead for each one, so we
  arrive at the familiar recursive function
  \begin{align}
    \label{eq:1}
    T(k) = (t_{k-1}+1)(T(k-1)+O(1)).
  \end{align}
  The base cases can be looked up in constant time and the constant
  term can be absorbed
, so we have $T(k)\in O(t_k)$.
\end{proof}

\section{Fault-Tolerant Hamiltonian Connectivity of DCell}
\label{sec:fault-tolerance}

A graph $G$ is called
$f$-fault Hamiltonian (resp. $f$-fault Hamiltonian-connected) if there
exists a Hamiltonian cycle (resp. if each pair of vertices are joined
by a Hamiltonian path) in $G\backslash F$ for any set $F$ of faulty
elements (faulty vertices and/or edges) with $|F| \leq f$.  For a
graph $G$ to be $f$-fault Hamiltonian (resp. $f$-fault
Hamiltonian-connected), it is necessary that $f \leq \delta(G) - 2$
(resp. $f \leq \delta(G)-3$), where $\delta(G)$ is the minimum degree
of $G$.

In this section we prove by induction on $k$, that $\dcell_k$ is
$(n+k-4)$-fault Hamiltonian-connected and $(n+k-3)$-fault Hamiltonian.
The base cases are given as Lemmas~\ref{lem:5}-\ref{lem:6}, and the main result is Theorem~\ref{thm:3}.

\begin{lemma}
  \label{lem:5}
$DCell_{0}$ is $(n-2)$-fault Hamiltonian-connected and
$(n-3)-$fault Hamiltonian.
\end{lemma}
\begin{proof}
 The lemma holds, since $DCell_{0}$ is a complete graph \cite{Diestel2012}.
\end{proof}

\begin{lemma}
  \label{lem:6}
$DCell_{2}$ with $n=2$ and $DCell_{1}$ with $n=3$ are $1$-fault Hamiltonian.
\end{lemma}
\begin{proof}
    This is verified by a computer program.
\end{proof}

\begin{theorem}
  \label{thm:3}
For any integer $n$ and $k$ with $n \geq 2$ and $k \geq 0$,
$DCell_{k}$ is $(n+k-4)$-fault Hamiltonian-connected and
$(n+k-3)$-fault Hamiltonian.
\end{theorem}
\begin{proof}
We will prove this theorem by induction on the dimension,
$k$, of $\dcell_k$.
The base cases are given in Lemmas~\ref{lem:5}--\ref{lem:6}.

Let $D_k=DCell_k$ and $D^\alpha_{k-1}$ denote a specific copy of $D_{k-1}$, with $\alpha \in \{0,1,\cdots,t_{k-1}\}$.
Our induction hypothesis is that $D_{k-1}$ is
$(n+k-5)$-fault Hamiltonian-connected and
$(n+k-4)$-fault Hamiltonian for $k> 0$ when $n\ge 3$, and $k>2$ when $n=2$.

Given a faulty set $F$ in $D_k$, our goal is to prove the following two results:

(1) $D_{k} \setminus F$ is Hamiltonian-connected if $|F| \leq n+k-4$;

(2) $D_{k} \setminus F$ is Hamiltonian if $|F| \leq n+k-3$.

\noindent For all $i = 0,1,\cdots,t_{k-1}$,
let $F_i$$ = F \cap D_{k-1}^{i}$, and $D^{i,f}_{k-1} = D^{i}_{k-1} \setminus F_i$,
the observation that we have the following four statements according to induction hypothesis:

(a) $D^{i,f}_{k-1}$ is Hamiltonian-connected if $|F_i| \leq n+k-5$ and $k > 1$;

(b) $D^{i,f}_{k-1}$ is Hamiltonian if $|F_i| \leq n+k-4$ and $k > 1$;

(c) $D^{i,f}_{0}$ is Hamiltonian-connected if $|F_i| \leq n-2$;

(d) $D^{i,f}_{0}$ is Hamiltonian if $|F_i| \leq n-3$.

Let $D'_{k} = D_{k} \setminus F$.
We can think of
$D'_{k}$ as a complete graph with the faulty set $F$, whose vertex set is the union of $D_{k-1}^{0,f}, D_{k-1}^{1,f}, \ldots,$ and $ D_{k-1}^{t_{k-1},f}$, and
whose edges are the level $k$ edges of $D'_{k}$.
Let $G=(V,E)$ be the
graph that is isomorphic to the complete graph $K_{t_{k-1}+1}$ with the faulty set $F$, with
$V=\{0,1,\ldots, t_{k-1}\}$, where vertex $i$ corresponds to
$D_{k-1}^{i,f}$, and edge $(i,j)$ corresponds to the level $k$ link that
connects $D_{k-1}^{i,f}$ to $D_{k-1}^{j,f}$ in $D'_{k}$.
Moreover, let $|F_\lambda| = $max$\{|F_0|, |F_1|, \cdots, |F_{t_{k-1}}|\}$ with $\lambda \in \{0,1,\cdots,t_{k-1}\}$.
We use $N(x,k)$ to denote the $k$ level neighbor of $x$ in $D_k$ with $k > 0$.

\textbf{Proof of (1).}
Let $u \in V(D_{k-1}^{\alpha,f})$ and $v \in V(D_{k-1}^{\beta,f})$ with $u \neq v$.
Then, we consider the following two cases.

Case 1.  $u,v$ are in the same copy of $D_{k-1}$.
We can claim the following four sub-cases.

Case 1.1.  $|F_\lambda| \leq n+k-5$.
By the induction hypothesis,  there exists a $(u,v)$-Hamiltonian path, $Q$,
in $D^{\alpha,f}_{k-1}$.
Furthermore, by Definition~\ref{def:dcell},
we have $t_{k-1} - |F_\alpha| - 1 \geq |F| - |F_\alpha| + 2$, thus,
there exist four distinct vertices $x,y,x',y'$
such that $(x,y) \in E(Q)$,
$x' = N(x,k) \in V(D_{k-1}^{\gamma,f})$, and
$y' = N(y,k) \in V(D_{k-1}^{\delta,f})$ with
$|P(Q,u,x)| < |P(Q,u,y)|$ and
$(x,x'),(y,y') \in E(D'_k)$.
Moreover, let $H_G$ be a Hamiltonian cycle in $G$, which contains the path
$(\gamma, \alpha, \delta)$, and let
$H$ be the corresponding set of level $k$ edges in $D'_{k}$.
By the
induction hypothesis, there is a Hamiltonian path in each $D_{k-1}^{i,f}$
for $i \neq \alpha$ whose first and last vertices are adjacent to
$k$-level edges in $H$. 
The union of these paths with $H$, $P(Q,u,x)$, and $P(Q,y,v)$, is a required $(u, v)$-Hamiltonian path (refer to
    Figure~\ref{fig:3.0}).

\begin{figure}
      \centering
      \def\svgwidth{\imgsize\textwidth}
{\small      \input{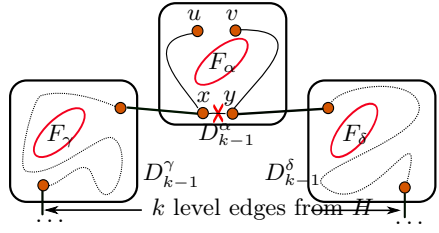}}
      \caption{Case 1.1 of the proof (1) in Theorem~\ref{thm:3}.}
      \label{fig:3.0}
\end{figure}

Case 1.2. $|F_\lambda| = n+k-4$ and $k=1$.
By Definition~\ref{def:dcell},
we have $t_0 -|F| -2 \geq 1$, thus,
there exist distinct vertices $x,x',$ and $y'$ such that
$x \in V(D_{0}^{\alpha,f} \setminus \{u,v\})$,
$x' = N(x,1) \in D_{0}^{\gamma,f}$, and $y' = N(v,1) \in D_{0}^{\delta,f}$.
Therefore, there exists a $(u,v)$-Hamiltonian path, $P$, in $D_{0}^{\alpha,f}$,
which contains the edge $(x,v)$.
Moreover, let $H_G$ be a Hamiltonian cycle in $G$, which contains the path
$(\gamma, \alpha, \delta)$, and let
$H$ be the corresponding set of level $1$ edges in $D'_{1}$.
By the
induction hypothesis there is a Hamiltonian path in each $D_{0}^{i,f}$
for $i \neq \alpha$ whose first and last vertices are adjacent to
$1$-level edges in $H$.
The union of these paths with $H$ and $P - (x,v)$, is a required $(u, v)$-Hamiltonian path.

Case 1.3. $|F_\alpha| = n+k-4$ and $k>1$.
By the induction hypothesis,  there exists a Hamiltonian cycle, $C$,
in $D^{\alpha,f}_{k-1}$.
We use ($u,u',\cdots,v,v',\cdots,u$) to denote $C$ if $(u,v)\notin E(C)$
and use ($u,\cdots,v'\neq u,v=u',u$) to denote $C$ if $(u,v)\in E(C)$.
What's more, let $x = N(v',k) \in V(D_{k-1}^{\gamma,f})$,
$y = N(u',k) \in V(D_{k-1}^{\delta,f})$, $P_{1} = P(C-u',u,v')$, and
\[
P_{2} =
\begin{cases}
\emptyset & \textrm{if} \ (u,v) \in E(C),  \\
P(C-u,u',v) & \textrm{otherwise}.  \\
\end{cases}
\]
Moreover, let $H_G$ be a Hamiltonian cycle in $G$, which contains the path
$(\gamma, \alpha, \delta)$, and let
$H$ be the corresponding set of level $k$ edges in $D'_{k}$.
By the
induction hypothesis there is a Hamiltonian path in each $D_{k-1}^{i,f}$
for $i \neq \alpha$ whose first and last vertices are adjacent to
$k$-level edges in $H$.
The union of these paths with $H$, $P_1$,  and $P_2$, is a required $(u, v)$-Hamiltonian path (refer to
    Figure~\ref{fig:3.1}).

\begin{figure}
      \centering
      \def\svgwidth{\imgsize\textwidth}
{\small      \input{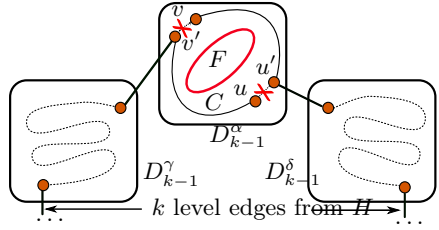}}
      \caption{Case 1.3 of the proof (1) in Theorem~\ref{thm:3} if $(u,v) \notin E(C)$.}
      \label{fig:3.1}
\end{figure}

Case 1.4. $|F_\lambda| = n+k-4$,  $\alpha \neq \lambda$, and $k>1$.
The case is similar to the Case 1.3, so we skip it.

Case 2. $u,v$ are in disjoint copies of $D_{k-1}$.
Then, we consider the following four sub-cases.

Case 2.1. $|F_\lambda| \leq n+k-5$.
By Definition~\ref{def:dcell},
we have $t_{k-1} - |F| - 2 \geq 1$, thus,
there exist distinct vertices $x,y,x',$ and $y'$,
such that $x \in V(D_{k-1}^{\alpha,f})$,
$y \in V(D_{k-1}^{\beta,f})$,
$x,y \notin \{u,v\}$,
$x' = N(x,k) \in V(D_{k-1}^{\gamma,f})$, and
$y' = N(y,k) \in V(D_{k-1}^{\delta,f})$ with
$x,y,x',$ and $y'$ are in disjoint copies of $D_{k-1}$ and
$(x,x'),(y,y') \in E(D'_k)$.
Moreover, let $H_G$ be a $(\alpha, \beta)$-Hamiltonian path in $G$, which contains the edge set
$\{(\alpha, \gamma),(\delta, \beta)\}$, and let
$H$ be the corresponding set of level $k$ edges in $D'_{k}$.
By the induction hypothesis, the union of $H$ with the Hamiltonian paths in each $D^{i,f}_{k-1}$
for $i = 0,1,\cdots,t_{k-1}$ is a required $(u, v)$-Hamiltonian path.

Case 2.2. $|F_\lambda| = n+k-4$ and $k=1$.
Choose $x,y,x',$ and $y'$ such that
$x \in V(D_{0}^{\alpha,f})$,
$y \in V(D_{0}^{\beta,f})$, $x,y \notin \{u,v\}$,
$x' = N(x,1) \in V(D_{0}^{\gamma,f})$, and
$y' = N(y,1) \in V(D_{0}^{\delta,f})$ with
$x,y,x',$ and $y'$ are in disjoint copies of $D_0$.
Moreover, let $H_G$ be a $(\alpha, \beta)$-Hamiltonian path in $G$, which contains the edge set
$\{(\alpha, \gamma),(\delta, \beta)\}$, and let
$H$ be the corresponding set of level $1$ edges in $D'_{1}$.
By the induction hypothesis, the union of $H$ with the Hamiltonian paths in each $D^{i,f}_{0}$
for $i = 0,1,\cdots,n$ is a required $(u, v)$-Hamiltonian path.

Case 2.3. $|F_\alpha| = n+k-4$ and $k>1$.
By the induction hypothesis,  there is a Hamiltonian cycle, $C$,
in $D^{\alpha,f}_{k-1}$.
Choose $x,y,x',$ and $y'$ such that $(u,x) \in E(C)$,
$y \in V(D_{k-1}^{\beta,f})$, $y \neq v$, $x' = N(x,k) \in D_{k-1}^{\gamma,f}$,
and $y' = N(y,k) \in V(D_{k-1}^{\delta,f})$ with
$x,y,x',$ and p$y'$ are in disjoint copies of $D_{k-1}$.
Thus, $C - (u,x)$ is a $(u,x)$-Hamiltonian path in $D_{k-1}^{\alpha,f}$.
Moreover, let $H_G$ be a $(\alpha, \beta)$-Hamiltonian path in $G$, which contains the edge set
$\{(\alpha, \gamma),(\delta, \beta)\}$, and let
$H$ be the corresponding set of level $k$ edges in $D'_{k}$.
By the induction hypothesis, the union of $H$ with the Hamiltonian paths in each $D^{i,f}_{k-1}$
for $i = 0,1,\cdots,t_{k-1}$ is a required $(u, v)$-Hamiltonian path.

Case 2.4. $|F_\lambda| = n+k-4$,  $\alpha \neq \lambda$, and $k>1$.
The case is similar to the Case 2.3, so we skip it.

\textbf{Proof of (2).}
We consider the following three cases with respect to $|F_\lambda|$.

Case 1. $|F_\lambda| \leq n+k-5$.
let $C_G$ be a Hamiltonian cycle in $G$, and let
$C$ be the corresponding set of level $k$ edges in $D'_{k}$.
By the induction hypothesis, the union of $C$ with the Hamiltonian paths in each $D^{i,f}_{k-1}$
for $i = 0,1,\cdots,t_{k-1}$ is a required Hamiltonian cycle.

Case 2. $|F_\lambda| = n+k-4$.
By the induction hypothesis, there is a Hamiltonian cycle, $C$, in $D^{\lambda,f}_{k-1}$.
By Definition~\ref{def:dcell},
we have $t_{k-1} - |F| -2 \geq 1$, thus,
there exist distinct vertices $u,v,u'$, and $v'$,
such that $(u,v) \in E(C)$,
$u' = N(u,k) \in V(D_{k-1}^{\delta,f})$, and $v' = N(v,k) \in V(D_{k-1}^{\gamma,f})$ with
$u',v' \in V(D'_{k})$ and $(u,u'),(v,v') \in E(D'_{k})$.
Therefore, $C-(u,v)$ is a $(u,v)$-Hamiltonian path, in $D^{\lambda,f}_{k-1}$.
Moreover, let $H_G$ be a Hamiltonian cycle in $G$,
which contains the path $(\gamma, \alpha, \delta)$,
and let $H$ be the corresponding set of level $k$ edges in $D'_{k}$.
By the induction hypothesis, the union of $H$ with the Hamiltonian paths in each $D^{i,f}_{k-1}$
for $i = 0,1,\cdots,t_{k-1}$ is a required Hamiltonian cycle (refer to
    Figure~\ref{fig:3.2}).

\begin{figure}
      \centering
      \def\svgwidth{\imgsize\textwidth}
{\small      \input{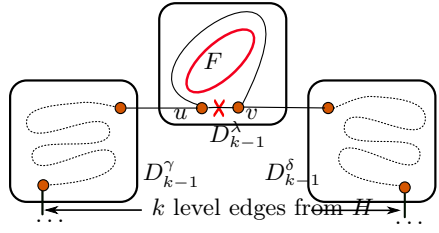}}
      \caption{Case 2 of the proof (2) in Theorem~\ref{thm:3}.}
      \label{fig:3.2}
\end{figure}

Case 3. $|F_\lambda| = n+k-3$. 
Choose an faulty element $x$ in $F_\lambda$.
Let $F'_\lambda = F_\lambda \setminus \{x\}$.
By the induction hypothesis, there is a Hamiltonian cycle, $C$, in $D^{\lambda}_{k-1}\setminus F'_\lambda$.
Thus, there is a $(u,v)$-Hamiltonian path, $C - x = (u,\cdots,v)$, in $D^{\lambda,f}_{k-1}$.
What's more, let $u' = N(u,k) \in V(D_{k-1}^{\gamma,f})$ and $v' = N(v,k) \in V(D_{k-1}^{\delta,f})$.
Moreover, let $H_G$ be a Hamiltonian cycle in $G$,
which contains the path $(\gamma, \alpha, \delta)$,
and let $H$ be the corresponding set of level $k$ edges in $D'_{k}$.
By the induction hypothesis, the union of $H$ with the Hamiltonian paths in each $D^{i,f}_{k-1}$
for $i = 0,1,\cdots,t_{k-1}$ is a required Hamiltonian cycle.

\end{proof}

\section{Incremental Expansion}
\label{sec:incremental}

A DCell can be deployed incrementally in a way that maintains high
connectivity at each step.  We show that a partial DCell, as described
in \cite{guo2008dcell}, is Hamiltonian connected if it conforms to a
few practical restrictions.  We also give a generalized and more
formal version of \adddcell~from \cite{guo2008dcell}.


Let $a_k,a_{k-1}, \ldots, a_2$ be positive integers with the property
that $a_2\le a_i$, for each $2<i\le k$.  Let $A=[a_k]\times
[a_{k-1}]\times \cdots \times [a_2]$, where $[a_i]$ denotes the set
$\{0,1,\ldots, a_i-1\}$.  A vertex of DCell is labeled by
$(\alpha_k,\alpha_{k-1},\ldots, \alpha_0)$, with
$0\le\alpha_i<t_{i-1}+1$, for $0< i\le k$, and $0\le \alpha_0<n$.
Recall that $\alpha_i$ represents the index of a $\dcell_{i-1}$ in a
$\dcell_{i}$, for $i>0$.  A partial DCell, however, uses $\dcell_1$ as
its unit of construction, and hence we index $A$ from 2, noting that
$a_2=t_1+1$.

\newcommand\listed{\phi}
\newcommand\nexte{\texttt{Next}}

We also formalize the operations of \adddcell~in
Algorithm~\ref{alg:adddcell} by defining the array $\listed$, indexed
by $A$.  It is initialized with $\listed(\alpha)=0$ for every $\alpha$
in $A$, and it stores the elements of a partial listing of $A$.  Given
$A$ and $\listed$, the call \nexte$(\varnothing)$ returns the next
element $\alpha$ of $A$ by setting $\listed(\alpha)\gets 1$.

Let $\alpha\in A$ such that $\alpha=\alpha_k\cdots\alpha_2$.  The
tuple $\alpha'$ is a \emph{prefix} of $\alpha$ if
$\alpha'=\alpha_k\cdots\alpha_{l+1}$ for some $2\le l+1 \le k$.  The
set of \emph{prefixes} of $A$ comprises the prefixes of the elements
of $A$.  If $A$ represents a $\dcell_k$, then $\alpha'$ is the unique
prefix of some sub-partial $\dcell_l$.

If $\listed(\alpha)=1$, then the prefixes of $\alpha$ are said to be
\emph{non-empty}, since this is when they contain a $\dcell_1$, and if
every $\alpha$ with some prefix $\alpha'$ has $\listed(\alpha)=1$, the
prefix $\alpha'$ is said to be \emph{full}, and this corresponds to
the full sub-DCell with prefix $\alpha'$.  Let $\varnothing$ denote the
prefix of length $0$.

\begin{algorithm}
  \caption{Given $A$ and $\listed$, computed by previous executions of
    \nexte$(\varnothing)$, calling \nexte$(\varnothing)$ computes an
    element $\alpha$ of $A$ and sets $\listed(\alpha)\gets 1$.  When
    $A$ represents a DCell, this is equivalent to
    \adddcell~in~\cite{guo2008dcell}.  Note that $\alpha'i$ denotes
    $\alpha_k\alpha_{k-1}\cdots\alpha_{l+1}i$.}
  \label{alg:adddcell}
  \begin{algorithmic}
    \Function{\nexte}{$\alpha'=\alpha_k\alpha_{k-1}\cdots\alpha_{l+1}$}%
    \State $m\gets \min\left\{i:\alpha' i \text{ is empty}\right\}$.
  \If{$l=2$}%
    \State $\listed(\alpha' m) \gets 1$.
    \State \Return.
    \EndIf%
    \If{$m \ge a_2$}
    \State $m\gets \min\left\{i:\alpha' i \text{ is not full}\right\}$.
    \EndIf
    \State \nexte$\left(\alpha' m\right)$.
    \EndFunction %
  \end{algorithmic}
\end{algorithm}

Algorithm~\ref{alg:adddcell} is more general than \adddcell, but it is
easy to see that it does enumerate the elements of $A$.  We describe
the order in which $A$ is enumerated below.


Let $L(A)$ be an ordered list of the elements of $A\setminus \{0\cdots
0\}$, which is defined recursively as follows: if $k=2$, then
$A=\{0,\ldots,a_2-1\}$ and $L(A) = 1,\ldots, a_2-1$.  Otherwise, let
$A'=[a_{k-1}]\times \cdots \times [a_2]$, and let $\lambda_i$ be the
$i$th tuple in the list $L(A')$ (of $A'\setminus \{0\cdots 0\}$).  Let
$mL(A')$ denote the ordered list whose $i$th element is equal to
$m\lambda_i$.  That is, we simply pre-pend $m$ to each element of
$L(A')$.  Note that we assume $0\le m < a_k$, so that $m\lambda_i\in
A$.  The listing of $A$, namely $0\cdots 0, L(A)$ is obtained by
concatenating the ordered lists given in Table~\ref{tab:LA}, which are
expressed using the above notation.  This listing is the order in
which the elements of $A$ are enumerated by $A$ executions of
$\nexte(\varnothing)$.

For example, take $A=[3]\times [3]\times [2]$.  We have $L([2])=(1)$,
and from the entries in Table~\ref{tab:LA}, $L([3]\times [2])$ is the
concatenation of the following $1$-element lists: $(10), (01), (11),
(20), (21)$.  So $L([3]\times [2]) = (10,01,11,20,21)$.  Now
$000,L(A)$ is given in Table~\ref{tab:LAx}, in the same format as
Table~\ref{tab:LA}.  The full list, therefore, is $000$, $100$, $010$,
$001$, $011$, $020$, $021$, $110$, $101$, $111$, $120$, $121$, $200$,
$210$, $201$, $211$, $220$, $221$.






  \begin{table}
    \centering
    \caption{The enumeration $0\cdots 0,L(A)$, of $A$, as generated by |A| executions of \nexte$(\varnothing)$.  The sub-lists are grouped according to the three main stages of $\nexte$.  See the main text for an explanation of the notation.} 
    \label{tab:LA}
$    \begin{array}{rlrl}
    0&(0\cdots 0)\\
    1&(0\cdots 0)\\
    \vdots&\vdots\\
    a_2-1&(0\cdots 0)\\
    \hline
    0 & L(A')\\
    \vdots&\vdots\\
    a_2-1 &  L(A')\\
    \hline
    a_2 & (0\cdots 0)\\
    a_2&L(A')\\
    a_2+1 & (0\cdots 0)\\
    a_2+1&L(A')\\
    \vdots&\vdots\\
    a_k-1 & (0\cdots 0)\\
    a_k-1&L(A')\\
  \end{array}
  $
  \end{table}

  \begin{table}
    \centering
    \caption{The enumeration of $[3]\times [3]\times [2]$ that is represented by Table~\ref{tab:LA}.}
    \label{tab:LAx}
    $
    \begin{array}{rl}
      0&(00)\\
      1&(00)\\
      \hline
      0&(10,01,11,20,21)\\
      1&(10,01,11,20,21)\\
      \hline
      2&(00)\\
      2&(10,01,11,20,21)\\
    \end{array}
    $
  \end{table}

  For the sake of practical implementation, we observe that the
  emptiness or fullness of any prefix can be verified with one query
  to $\listed$.

\begin{corollary}
  \label{cor:nexte}
  Let $\alpha'=\alpha_k\cdots\alpha_{l+1}$ be a prefix of $A$.  We
  have the following two facts
  \begin{enumerate}
  \item The prefix $\alpha'$ is non-empty if and only
    if
    \begin{align*}
      \listed(\alpha'0\cdots 0)=1.
    \end{align*}
  \item The prefix $\alpha'$ is full if and only if
    \begin{align*}
    \listed(\alpha'(a_l-1)\cdots(a_2-1))=1.
    \end{align*}
  \end{enumerate}
\end{corollary}
\begin{proof}
  This can be seen in Table~\ref{tab:LA}.
\end{proof}

Let $A$ be the set of prefixes of $\dcell_1$s in a DCell for some $n$
and $k$.  A \emph{partial DCell}, denoted $D_{d,k}$, consists of the
$\dcell_1$-prefixes in the pre-image $\listed^{-1}(1)$ after $d$ calls
to $\nexte(\varnothing)$, and any links that can be added using the
connection rule of Definition~\ref{def:dcell}.  The unique sub-partial
$\dcell_l$ of $D_{d,k}$, with prefix $\alpha_k\alpha_{k-1}\cdots
\alpha_{l+1}$, is denoted by $D_{d,l}^{\alpha_k\alpha_{k-1}\cdots
  \alpha_{l+1}}$.  Note that $D_{d,l}^{\alpha_k\alpha_{k-1}\cdots
  \alpha_{l+1}}$ may be empty.

The following lemma holds for general $A$.



\newcommand\cc{d}

\begin{lemma}
  \label{lem:partial}
  Let $\alpha'$ be a prefix in $A$, and let $0<\cc<a_2-1$.  If the
  prefix $\alpha'd$ is non-empty, then the prefix $\alpha(d-1)$ is
  non-empty and the prefix $\alpha'(d+1)$ becomes non-empty after at
  most one call to \nexte$(\varnothing)$.
\end{lemma}
\begin{proof}
  This can be seen in Table~\ref{tab:LA}.

\end{proof}

\begin{definition}
  Let $\listed$ be a partial listing of $A$ and let $0<c\le a_2$.  We
  say that the partial listing $\listed$ is \emph{$K_c$-connected} if
  the following holds for every prefix $\alpha'$:

If $\alpha'1$ is non-empty, then $\alpha'(c-1)$ is also non-empty.
\end{definition}

Further on, we require that a Hamiltonian connected partial DCell be
$K_c$-connected for certain $c$, but Corollary~\ref{cor:connected}
 shows that this is not an unreasonable condition.

\begin{corollary}
  \label{cor:connected}
  Any $\listed$ will become $K_c$-connected after at most $c-2$
  further calls to \nexte$(\varnothing)$.
\end{corollary}
\begin{proof}
  This follows from Lemma~\ref{lem:partial}.
\end{proof}

\newcommand\uinv{\psi}
\newcommand\flr[2]{\left\lfloor \frac{#1}{#2}\right\rfloor}

The high connectivity of a partial DCell comes from
Lemma~\ref{lem:partial} and the connection rule of
Definition~\ref{def:dcell}, and we use it to prove a stronger version
of Theorem~7 of \cite{guo2008dcell}.



\begin{lemma}
  \label{lem:thm7} Let $D_{d,k}$ be a partial DCell, and let $\alpha'
  = \alpha_k\cdots\alpha_{l+2}$, with $3\le l+2 \le k$.  
  If $m$ is the largest integer such that $D_{d,l}^{\alpha'm}$ is
  non-empty, then $D_{d,l}^{\alpha'i}$ and $D_{d,l}^{\alpha'j}$ are
  linked for all $i$ and $j$ such that $0\le i<j< m$, and
  $D_{d,l}^{\alpha'm}$ itself is linked to at least $\min\{m,t_1\}$
  sub-partial $\dcell_l$s.

\end{lemma}

\begin{proof}
  The first part is exactly the statement of Theorem~7 in
  \cite{guo2008dcell}, since removing $D_{d,l}^{\alpha'm}$ yields a
  completely connected $D_{d,l+1}^{\alpha'}$.  For the second part,
  observe that each $D_{d,l}^{\alpha'i}$, for $0\le i \le m$ contains
  at least $D_{d,1}^{\alpha'i0\cdots 0}$, which has $t_1$ servers, and
  if $m\ge t_1+1 (=a_2)$, then $D_{d,l}^{\alpha'i}$ is full for $i<m$.
  Thus, by the connection rule of Definition~\ref{def:dcell},
  $D_{d,1}^{\alpha'm0\cdots 0}$ is linked to $D_{d,l}^{\alpha'i}$ for
  $0\le i < \min\{m,t_1\}$, as required (see Figure~\ref{fig:thm7}).
\end{proof}

\begin{figure}
 \centering
 \input{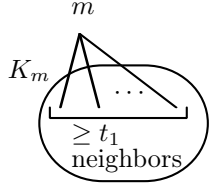}
 \caption{Lemma~\ref{lem:thm7} describes a graph isomorphic to
   $K_{m+1}$ if $m\le t_1$, and a graph isomorphic to the one shown
   if $m>t_1$.}

 \label{fig:thm7}
\end{figure}

We show that certain partial DCells are Hamiltonian connected,
however, the $n$ and $k$ parameters that do not satisfy the
antecedents of Theorem~\ref{thm:partial} are of little consequence, in
practical terms.  Corollary~\ref{cor:connected} shows that
$K_c$-connectedness can be obtained by adding at most $c-2$ more
$\dcell_1$s.  For small $n$, say $n\le 5$, we may use
Theorem~\ref{thm:partial} (b), and for larger $n$ we may use (a).  A
property slightly stronger than Hamiltonian connectivity is proven.

\begin{theorem}
  \label{thm:partial}
  Let $D_{d,k}$ be a $K_c$-connected partial DCell such that $4\le
  n\le c-1 < t_1+1 (=a_2)$; and,
  \begin{enumerate}[(a)]
  \item either $ c<t_1+1$ and $k+1\le n$ and $\omega=0$; or,
  \item  $c=t_1+1$ and $k+1 \le t_1$ and $\omega=1$.
  \end{enumerate}
  There is a $(u,v)$-Hamiltonian path, $H = (u=v_0,v_1, \ldots,
  v_{t_k-1}=v)$, such that at most the first $k-1$ vertices of $H$ that
  occur in $D_{d,\omega}^{0\cdots 0}$ are not consecutive in $H$.

\end{theorem}
\begin{proof}
  The proof is similar to Theorem~\ref{thm:1}, but we must be more
  careful with Case 1, because not every vertex is incident to a level
  $k$ link.  We proceed by induction on $k$.

  Clearly $D_{d,1}$ is Hamiltonian connected, for either $d=0$ and it
  is empty, or $d=1$ and it is Hamiltonian connected by
  Theorem~\ref{thm:1}.  Suppose that the theorem holds for all $1\le
  k'<k$.

  Let $D_{d,k}$ be a partial DCell satisfying the antecedents of the
  theorem statement, and let $u$ and $v$ be distinct vertices of
  $D_{d,k}$.  We combine Hamiltonian paths in its non-empty partial
  $\dcell_{k-1}$s with a Hamiltonian cycle on its level $k$ edges.
  There are two cases.

  Case 1, $u$ and $v$ are in the same partial $\dcell_{k-1}$.
  Let $u, v \in V(D_{d,k-1}^{\alpha_k})$, for some $\alpha_k$.  By the
  inductive hypothesis, there is a $(u,v)$-Hamiltonian path, $P=
  (u=v_0,v_1, \ldots, v_{t_{k-1}-1}=v)$, in $D_{d,k-1}^{\alpha_k}$,
  such that at most the first $k-1$ vertices of $P$ that occur in
  $D_{d,\omega}^{\alpha_k 0\ldots 0}$ are not consecutive in $P$.  Let
  $x$ and $y$ be the $k$th and $(k+1)$st such vertices, so that
  $(x,y)\in E(P)\cap E\left(D_{d,\omega}^{\alpha_k 0\ldots 0}\right)$.

  Since $D_{d,k-1}^{\alpha_k}$ contains $u$ and $v$,
  Corollary~\ref{cor:nexte} says that $D_{d,\omega}^{\alpha_k 0\cdots
    0}$ is non-empty, and we must show that the respective level $k$
  neighbors of $x$ and $y$ exist in $D_{d,k}$.  We do this by showing
  that every vertex in $D_{d,\omega}^{\alpha_k 0\cdots 0}$ has a level
  $k$ neighbor in $D_{d,k}$. 

    Let $w$ be a vertex of $D_{d,\omega}^{\alpha_k 0\cdots 0}$, and
    let $i=uid_{k-1}(w)$ (see Definition~\ref{def:dcell}).  Suppose
    $i\ge \alpha_k$, then we must show that the vertex with
    $uid_{k-1}$ equal to $\alpha_k$ in $D_{d,k-1}^{i+1}$ exists in
    $D_{d,k}$.  There are $t_{\omega}$ vertices in
    $D_{d,\omega}^{\alpha_k0\cdots 0}$ with $uid_{k-1}$s
    $0,1,\ldots,t_{\omega}-1$, so we have $i<t_{\omega}$, and thus
    $i+1 \le t_{\omega}<c$.  The rightmost inequality follows from the
    antecedents of the theorem.  By $K_c$-connectivity,
    $D_{d,k-1}^{i+1}$ is non-empty and, by Corollary~\ref{cor:nexte},
    it contains $D_{d,1}^{(i+1)0\cdots 0}$ with the vertex $uid_{k-1}$
    $\alpha_k$.

    Now suppose $i<\alpha_k$, so we must show that the vertex with
    $uid_{k-1}$ $\alpha_k-1$ in $D_{d,k-1}^{i}$ exists in $D_{d,k}$.
    By the same argument as above, $D_{d,k-1}^{i}$ is non-empty, and
    likewise if $\alpha_k-1<t_1$, we can find this vertex in
    $D_{d,1}^{i0\cdots 0}$.  In the case where $\alpha_k-1\ge t_1$,
    the implication that $D_{d,k-1}^{\alpha_k}$ is non-empty (it
    contains $i$), and $\alpha_k\ge t_1+1( = a_2)$ is that
    $D_{d,k-1}^{i}$ must be full, so the vertex with $uid_{k-1}$
    $\alpha_k-1$ must be present in $D_{d,k-1}^{i}$, as required.






    The two sub-paths of $P$, from $u$ to $x$ and from $y$ to $v$ are
    combined with a Hamiltonian cycle on the level $k$ edges and
    Hamiltonian paths in each of the other sub-partial $\dcell_{k-1}$s
    to form a $(u,v)$-Hamiltonian path, $H$.  It remains to show that
    the level $k$ Hamiltonian cycle exists.


    In practice, the cycle should be easy to find, but for the sake of
    this proof we use the Bondy-Chv\'atal theorem
    \cite{BondyChvatal1976} to ensure its existence.
    Let $E$ be the set of level $k$ edges in $D_{d,k}$, minus the
    edges incident with vertices in $V(D_{d,k-1}^{\alpha_k})\setminus
    \{x,y\}$.  Let $G=(V,E)$ be the graph on $m+1$ vertices,
    representing the $\dcell_{k-1}$s in $D_{d,k}$, with
    $V=\{0,1,\ldots, m\}$.  If $G$ is Hamiltonian, then the cycle
    required for Case 1 exists.

    Denote the degree of vertex $i$ by $d(i)$.  The \emph{closure} of
    $G$ is the graph obtained by repeatedly adding an edge between
    non-adjacent vertices $i$ and $j$ whenever $d(i)+d(j)\ge
    |V(G)|(=m+1)$, until no more edges can be added. The
    Bondy-Chv\'atal theorem states that $G$ is Hamiltonian if and only
    if its closure is Hamiltonian.

    If $\alpha_k = m$ and $i<m$, then $d(\alpha_k)=2$ and by
    Lemma~\ref{lem:thm7}, we have $d(i)\ge m-1$.  The closure of $G$
    in this case is $K_{m+1}$, so $G$ is Hamiltonian.

    If $\alpha_k<m$ and $i<m$ and $i\neq \alpha_k$, then
    $d(\alpha_k)=2$ and by Lemma~\ref{lem:thm7} we have $d(m) \ge
    \min(t_1,m)-1$, and finally, $d(i)\ge m-2$.  Now by
    $K_c$-connectivity, $m+1\ge c\ge n+1\ge 5$, so we reduce to the
    above case (where $\alpha_k=m$), by remarking that $m-2+m-1 \ge
    m+1$ and $m-2+t_1-1\ge m+1$.  Thus $G$ is Hamiltonian.

    A $(u,v)$-Hamiltonian path, $H$, can now be constructed using the
    arguments in Case 1 of Theorem~\ref{thm:1}, and it remains to show
    that $H$ satisfies the stronger requirements of the present
    theorem.

    If $\alpha_k= 0$, then the theorem is satisfied by the above
    construction, and if $\alpha_k\neq 0$, then by the inductive
    hypothesis, at most the first $k-2$ vertices of $H$ that occur in
    $D_{d,\omega}^{0\cdots 0}$ are not-consecutive (on $H$).  Thus the
    theorem is also satisfied.



    \def\svgwidth{0.35\textwidth}
   \begin{figure}[h]
     \centering
     \input{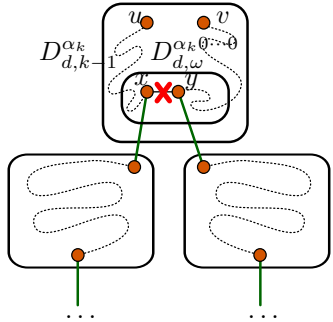}
     \caption{Case 1 for partial DCell.}
     \label{fig:case1partial}
   \end{figure}

   Case 2, $u$ and $v$ are in distinct sub-partial
     $\dcell_{k-1}$s.  In Theorem~\ref{thm:1} we used two subcases
   for clarity of argument, but we avoid this now for brevity's sake.
   Let $u\in V(D_{d,k-1}^{\alpha_k})$ and $v\in
   V(D_{d,k-1}^{\beta_k})$ with $\alpha_k\neq \beta_k$, and let $u$
   and $v$ be connected by level $k$ edges to vertices in
   $D_{d,k-1}^{\gamma_k}$ and $D_{d,k-1}^{\delta_k}$, respectively, if
   they exist.  Let $E$ be the set of level $k$ edges of $D_{d,k}$.
   Let $G=(V,E)$ be the graph whose $m+1$ vertices represent the
   $\dcell_{k-1}$s of $D_{d,k}$, with $V=\{0,\ldots, m\}$.  Let $G'$
   be $G$ plus a vertex, called $t$, of degree $2$, connected to
   vertices $\alpha_k$ and $\beta_k$.

    Note that in the case that both $D_{d,k-1}^{\gamma_k}$ and
    $D_{d,k-1}^{\delta_k}$ exist, there are three mutually exclusive
    possibilities.  Either $\gamma_k,\delta_k,\alpha_k$ and $\beta_k$
    are all distinct, or $\gamma_k=\delta_k$, or
    $(\alpha_k,\gamma_k)=(\delta_k,\beta_k)$.

    In any event, let $G''$ be the graph $G'$ minus the edge set
    $\{(\alpha_k,\gamma_k),(\delta_k,\beta_k)\}$ (if these edges exist).
    Case 2 holds if $G''$ is Hamiltonian, so once again we use the
    Bondy-Chv\'atal theorem \cite{BondyChvatal1976} to ensure this,
    by showing that the closure of $G''$ is $K_{m+2}$.

    By Lemma~\ref{lem:thm7}, the degrees of vertices of $G$ satisfy
    $d_{G}(m)\ge \min(t_1,m)$ and $m-1\le d_{G}(i) \le m$, for $0\le i
    < m$, and in particular, $d_{G}(i) = m$ if $m\le t_1$.  If this is
    the case, then the graph $G''$ is $K_{m+1}$ plus the vertex $t$
    and the edges $\{(t,\alpha_k), (t,\beta_k)\}$, minus the edges
    $\{(\alpha_k,\gamma_k),(\delta_k,\beta_k)\}$, if they exist.  If
    they do not exist, we are done, since the closure of such a graph
    is $K_{m+2}$, so suppose that they do exist.  Without loss of
    generality, we need only show that $(\alpha_k,\gamma_k)$ can be
    added back, and there are three cases to consider, recalling that
    $m\ge 4$ (see Figure~\ref{fig:case2partial}).  If $\gamma_k=\beta_k$, then
    \begin{align*}
      &d_{G''}(\alpha_k)+d_{G''}(\gamma_k) \\
      =& m+m=2m \ge m+2.
    \end{align*}
    If $\gamma_k=\delta_k$, then
    \begin{align*}
      &d_{G''}(\alpha_k)+d_{G''}(\gamma_k) \\
      =& m+m-2=2m-2 \ge m+2.
    \end{align*}
    If $\gamma_k$ is not equal to either of these, then
    \begin{align*}
      &d_{G''}(\alpha_k)+d_{G''}(\gamma_k) \\
      =& m + m-1=2m-1\ge m+2.
    \end{align*}

   \begin{figure} 
     \centering
     \input{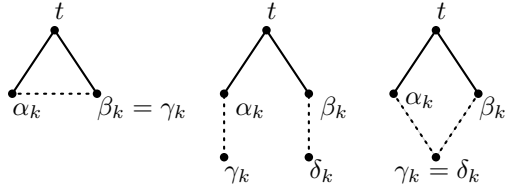}
     \caption{Either one or two dashed edges are removed to obtain $G''$
       from $G'$, and the three cases are depicted, from \emph{left}
       to \emph{right}.}
     \label{fig:case2partial}
   \end{figure}
    If $m>t_1$, then we need not be so precise about the degrees to
    achieve the result.  Recall that $t_1 = n(n+1) \ge 20$, and notice
    that if $i$ and $j$ satisfy $0\le i<j\le m$, then
    $d_{G}(i)+d_{G}(j)\ge m-1 + t_1$.  In removing
    $\{(\alpha_k,\gamma_k),(\delta_k,\beta_k)\}$, these degrees may be
    reduced so that $d_{G''}(i)+d_{G''}(j)\ge m-1 + t_1 - 2$, but
    $t_1-3> 2$, so the closure of $G''$ is $K_{m+2}$.



    Once again, we proceed as in Case 2 of Theorem~\ref{thm:1},
    combining a Hamiltonian path on the level $k$ edges with the
    required Hamiltonian paths in each sub-partial $\dcell_{k-1}$.  By
    the same argument of conclusion to the previous case, our
    $(u,v)$-Hamiltonian path satisfies the requirements of the
    theorem.
\end{proof}

Our proof for partial DCell does not extend to partial generalized
DCell, since its construction is specific to one connection rule.  On
the other hand, Algorithm~\ref{alg:dcellhp} can be readily adapted to
operate on a partial DCell with similar (perhaps equal) time
complexity.

\section{Discussions}
\label{sec:discussion}

We discuss some of the remaining aspects of implementing a Hamiltonian
cycle broadcast protocol, such as load balancing, latency, and an
alternative view of fault tolerance. In a fault-free DCell, a
Hamiltonian cycle is computed using \DCellHP, and the appropriate
forwarding information is sent to each node of the network.  A
corresponding identification is incorporated into any packet we wish
to broadcast over the Hamiltonian cycle, so that DCell's forwarding
module (\cite{guo2008dcell}) will find the next vertex of the cycle in
a constant amount of time at each step.

For load balancing purposes, several different Hamiltonian cycles can
be computed and their forwarding information stored across the
network.  We leave open the issue of choosing the best combination of
cycles, and choosing which one to broadcast over.

Broadcasting over a Hamiltonian cycle in a DCell exchanges speed for
efficiency, but DCellHP can be used to combine Hamiltonian cycles in
many sub-DCells, in order to reach all vertices of the network sooner.
For example, we can run DCellHP on each $\dcell_{k-1}$, and broadcast
within each of these by routing over a Hamiltonian cycle in
$D_{k-1}^0$, and branching along the level $k$ edge joining
$D_{k-1}^0$ to $D_{k-1}^i$, for each $0<i\le t_{k-1}$.  This way the
broadcast finishes in $O(t_{k-1})$ time, plus a small amount of time
taken to send the packet to the start node in $D_{k-1}^0$.

Most of the subtleties, however, arise when network faults are
introduced.  Intuitively, a large DCell network has many Hamiltonian
cycles which can be found using different choices for the $k$-level
Hamiltonian cycle $H$ in \DCellHP.  Thereby, a certain number of
$k$-level link faults can be avoided, as we have shown in
Section~\ref{sec:fault-tolerance}.  The aforementioned discussion
assumes that the faults are chosen by an adversary, who may place them
in the worst possible locations.  There is another, equally important
discussion to be had about randomly distributed faults.  That is,
given a uniform failure rate of $p$, for some $0\le p \le 1$, with
what probability does (a modified) DCellHP 
find a Hamiltonian cycle?
We leave this question open.

\section{Concluding Remarks}
\label{sec:conclusion}


Our primary goal in this research is to provide an alternative way of
broadcasting messages in DCell.  Toward this end, we have shown that
(almost all) Generalized DCell and several related graphs are
Hamiltonian connected.  Perhaps equally important, however, is opening
up a mathematical discussion on server-centric data center networks.
Which of these networks is Hamiltonian?

The answer is certainly yes for BCube \cite{GuoLuLi2009}, but for
others it is less clear.  Consider FiConn \cite{LiGuoWu2009}, for
example, whose construction is similar to DCell's.  That is, a level
$k$ FiConn is built by completely interconnecting a number of level
$k-1$ FiConns.  One might intuit that FiConn is Hamiltonian connected,
and that our proof for DCell can be adapted for showing this, however,
FiConn has the important distinction that each vertex has at most one
level $i$ link for $i> 0$.  This causes the induction step that we use
for DCell to fail.  It is easy to show that FiConn$_{n,1}$ is
Hamiltonian connected and that every FiConn$_{n,2}$ is Hamiltonian for
even $n$, where $n\ge 4$, however, the Hamiltonicity of FiConn in
general remains open.

\section*{Acknowledgements} This work is supported by National Natural
Science Foundation of China (No. 61170021), Application Foundation
Research of Suzhou of China (No.  SYG201240), Graduate Training
Excellence Program Project of Soochow University (No. 58320235), and
Natural Science Foundation of the Jiangsu Higher Education
Institutions of China (No. 12KJB520016), also by the Engineering and
Physical Sciences Research Council (EPSRC Reference EP/K015680/1), in
the United Kingdom.

\bibliographystyle{plain}
\bibliography{bibliography}

\begin{thebibliography}{10}

\bibitem{al2008scalable}
Mohammad Al-Fares, Alexander Loukissas, and Amin Vahdat.
\newblock A scalable, commodity data center network architecture.
\newblock In {\em ACM SIGCOMM Computer Communication Review}, volume~38, pages
  63--74. ACM, 2008.

\bibitem{BondyChvatal1976}
John~Adrian Bondy and V~Chv{\'a}tal.
\newblock A method in graph theory.
\newblock {\em Discrete Mathematics}, 15(2):111--135, 1976.

\bibitem{dean2008mapreduce}
Jeffrey Dean and Sanjay Ghemawat.
\newblock Mapreduce: simplified data processing on large clusters.
\newblock {\em Communications of the ACM}, 51(1):107--113, 2008.

\bibitem{Diestel2012}
Reinhard Diestel.
\newblock {\em Graph Theory, 4th Edition}, volume 173 of {\em Graduate texts in
  mathematics}.
\newblock Springer, 2012.

\bibitem{fan2002hamilton}
Jianxi Fan.
\newblock Hamilton-connectivity and cycle-embedding of the {M}{\"o}bius cubes.
\newblock {\em Information Processing Letters}, 82(2):113--117, 2002.

\bibitem{garey1979computers}
Michael~R Garey and David~S Johnson.
\newblock Computers and intractability: a guide to np-completeness, 1979.

\bibitem{ghemawat2003google}
Sanjay Ghemawat, Howard Gobioff, and Shun-Tak Leung.
\newblock The google file system.
\newblock In {\em ACM SIGOPS Operating Systems Review}, volume~37, pages
  29--43. ACM, 2003.

\bibitem{GuoLuLi2009}
Chuanxiong Guo, Guohan Lu, Dan Li, Haitao Wu, Xuan Zhang, Yunfeng Shi, Chen
  Tian, Yongguang Zhang, and Songwu Lu.
\newblock B{C}ube: a high performance, server-centric network architecture for
  modular data centers.
\newblock {\em SIGCOMM Comput. Commun. Rev.}, 39(4):63--74, August 2009.

\bibitem{guo2008dcell}
Chuanxiong Guo, Haitao Wu, Kun Tan, Lei Shi, Yongguang Zhang, and Songwu Lu.
\newblock D{C}ell: a scalable and fault-tolerant network structure for data
  centers.
\newblock In {\em ACM SIGCOMM Computer Communication Review}, volume~38, pages
  75--86. ACM, 2008.

\bibitem{isard2007dryad}
Michael Isard, Mihai Budiu, Yuan Yu, Andrew Birrell, and Dennis Fetterly.
\newblock Dryad: distributed data-parallel programs from sequential building
  blocks.
\newblock {\em ACM SIGOPS Operating Systems Review}, 41(3):59--72, 2007.

\bibitem{johnson1982np}
David~S Johnson.
\newblock The np-completeness column: An ongoing gulde.
\newblock {\em Journal of Algorithms}, 3(4):381--395, 1982.

\bibitem{kliegl2010generalized}
Markus Kliegl, Jason Lee, Jun Li, Xinchao Zhang, Chuanxiong Guo, and David
  Rinc{\'o}n.
\newblock Generalized {D}{C}ell structure for load-balanced data center
  networks.
\newblock In {\em INFOCOM IEEE Conference on Computer Communications Workshops,
  2010}, pages 1--5. IEEE, 2010.

\bibitem{KlieglLeeLi2009}
Markus Kliegl, Jason Lee, Jun Li, Xinchao Zhang, David Rincon, and Chuanxiong
  Guo.
\newblock The generalized dcell network structures and their graph properties.
\newblock Microsoft Research, October 2009.

\bibitem{LiGuoWu2009}
Dan Li, Chuanxiong Guo, Haitao Wu, Kun Tan, Yongguang Zhang, and Songwu Lu.
\newblock Fi{C}onn: Using backup port for server interconnection in data
  centers.
\newblock In {\em INFOCOM}, pages 2276--2285, 2009.

\bibitem{LinLiuHamdi2012}
Dong Lin, Yang Liu, Mounir Hamdi, and Jogesh~K. Muppala.
\newblock Hyper-bcube: A scalable data center network.
\newblock In {\em ICC}, pages 2918--2923. IEEE, 2012.

\bibitem{lin1994deadlock}
Xiaola Lin, Philip~K. McKinley, and Lionel~M. Ni.
\newblock Deadlock-free multicast wormhole routing in 2-d mesh multicomputers.
\newblock {\em Parallel and Distributed Systems, IEEE Transactions on},
  5(8):793--804, 1994.

\bibitem{park2005fault}
Jung-Heum Park, Hee-Chul Kim, and Hyeong-Seok Lim.
\newblock Fault-hamiltonicity of hypercube-like interconnection networks.
\newblock In {\em Parallel and Distributed Processing Symposium, 2005.
  Proceedings. 19th IEEE International}, pages 60a--60a. IEEE, 2005.

\bibitem{wang2012hamiltonian}
Dajin Wang.
\newblock Hamiltonian embedding in crossed cubes with failed links.
\newblock {\em Parallel and Distributed Systems, IEEE Transactions on},
  23(11):2117--2124, 2012.

\bibitem{wang2005multicast}
Nen-Chung Wang, Cheng-Pang Yen, and Chih-Ping Chu.
\newblock Multicast communication in wormhole-routed symmetric networks with
  hamiltonian cycle model.
\newblock {\em Journal of Systems Architecture}, 51(3):165--183, 2005.

\bibitem{xu2013hamiltonian}
Dacheng Xu, Jianxi Fan, Xiaohua Jia, Shukui Zhang, and Xi~Wang.
\newblock Hamiltonian properties of honeycomb meshes.
\newblock {\em Information Sciences}, 240:184--190, 2013.

\bibitem{yang2011hamiltonian}
Xiaofan Yang, Qiang Dong, Erjie Yang, and Jianqiu Cao.
\newblock Hamiltonian properties of twisted hypercube-like networks with more
  faulty elements.
\newblock {\em Theoretical Computer Science}, 412(22):2409--2417, 2011.

\end{thebibliography}

\end{document}